\documentclass[%
 jmp,
 amsmath,amssymb,
 reprint,%
]{revtex4-2}

\usepackage{graphicx}% Include figure files
\usepackage{dcolumn}% Align table columns on decimal point
\usepackage{bm}% bold math
\usepackage[utf8]{inputenc}
\usepackage[T1]{fontenc}
\usepackage{mathptmx}
\usepackage{etoolbox}

\newtheorem{definition}{Definition}
\newtheorem{theorem}{Theorem}
\newtheorem{remark}{Remark}
\newenvironment{proof}[1]{\textbf{Proof.} }{$\blacksquare$ \vspace{5pt}}

\makeatletter
\def\@email#1#2{%
 \endgroup
 \patchcmd{\titleblock@produce}
  {\frontmatter@RRAPformat}
  {\frontmatter@RRAPformat{\produce@RRAP{*#1\href{mailto:#2}{#2}}}\frontmatter@RRAPformat}
  {}{}
}%
\makeatother

\def\eqa#1{\begin{equation}\begin{aligned}#1\end{aligned}\end{equation}}

\def\seqs#1{\begin{equation*}\begin{split}#1\end{split}\end{equation*}}

\begin{document}

%\preprint{AIP/123-QED}
\title{Meromorphy of solutions for a wide class of ordinary differential equations of Painlev\'e type}

\author{A. V. Domrin}
 \homepage{email: domrin@mi-ras.ru}
 %\altaffiliation[Also at ]{Math inst.}%Lines break automatically or can be forced with \\
\affiliation{Faculty of Mechanics and Mathematics, Moscow State University, 
GSP-2, Leninskie Gory, Moscow 119992, Russian Federation;\\ Institute of Mathematics, Ufa Federal Research Centre, Russian Academy 
of Sciences, 112, Chernyshevsky Street, Ufa 450008, Russian Federation;\\ Moscow Centre of Fundamental and Applied Mathematics of Moscow State University}%\\This line break forced with \textbackslash\textbackslash}%

\author{M. A. Shumkin} %\email{shumkin.mikhail@yandex.ru} 
\homepage{email: shumkin.mikhail@yandex.ru}
\affiliation{Faculty of Mechanics and Mathematics, Moscow State University, 
GSP-2, Leninskie Gory, Moscow 119992, Russian Federation;\\ Moscow Centre of Fundamental and Applied Mathematics of Moscow State University}%\\This line break forced with \textbackslash\textbackslash}%

\author{B. I. Suleimanov}%
\homepage{email: bisul@mail.ru}
%\homepage{http://www.Second.institution.edu/~Charlie.Author.}
\affiliation{Institute of Mathematics, Ufa Federal Research Centre, Russian Academy 
of Sciences, 112, Chernyshevsky Street, Ufa 450008, Russian Federation}%\\This line break forced% with \\}%

\date{\today}% It is always \today, today,
             %  but any date may be explicitly specified

\begin{abstract}We prove the meromorphy of solutions for a wide class of ordinary
differential equations. These equations are given by invariant manifolds of non-linear partial
differential equations integrable by the inverse scattering method. Some higher analogues
of the Painlev\' e equations are considered as examples.\end{abstract}
%\bigskip
%\noindent MSC: 35L70, 35Q51, 35Q75,  53A30,  53Z05.
%\bigskip
%\noindent {\bf Keywords:} Painlev\' e equations, meromorphy, integrability, symmetry, invariant ma\-ni\-folds.

\maketitle

%\eqs{line1 \\ line2\\line3\\line4}
%\eqn{line1 \\ line2\\line3\\line4}
%\eqa{line1 \\ line2\\line3\\line4}
%\seqs{line1 \\ line2\\line3\\line4}
%\eq{line1}
%\seq{line1}

\section{Introduction }%\label{sec:intro}

 The six Painlev\' e equations are distinguished among the other second-order
ordinary differential equations (ODE) $w''_{zz}=f(z,w,w'_z)$ whose right-hand side is a rational
function of $w$ and $w'_z$, by the absence of non-polar movable singularities of their solutions
(that is, singularities whose position depends on the initial data).
In the simplest cases of the first Painlev\' e equation
\begin{equation}\label{PI}
w''_{zz}=6w^2+z
\end{equation}
and the second Painlev\' e equation
\begin{equation}\label{PII}
w''_{zz}=2w^3+zw+\alpha,
\end{equation}
this absence is equivalent to the property of global meromorphy of all their solutions with
arbitrary initial conditions
\begin{equation}\label{Ko}
w(z_0)=a, \qquad w'_{z}(z_0)=b,
\end{equation}
where $z_0$, $a$ and $b$ are any complex constants.

The Painlev\' e ODE are currently being applied to a wide variety of problems in mathematics 
and mathematical physics. But the problem of rigorously proving the absence of non-polar
movable singularities of their solutions turned out to be difficult. Original proofs by 
Painlev\' e and his followers appeared to be incomplete.  Gaps were also found in a number 
of subsequent attempts to give a satisfying proof. The meromorphy of all solutions of the
simplest Painlev\' e ODE (\ref{PI}) and (\ref{PII}) with initial conditions (\ref{Ko}) was 
rigorously justified for the first time only in~1999 by Hinkkanen and Laine~\cite{Laine}. 
A glimpse to the rather dramatic story of proving the absence of non-polar movable 
singularities of solutions of the Painlev\' e equations can be found in the introduction 
of the recent paper~\cite{dobish}, along with references to the literature.

A principally new proof of the global meromorphy of solutions of
(\ref{PI}) and (\ref{PII}) was also given in~\cite{dobish}. It is based on
the well-known fact that these ODE describe self-similar solutions 
of the Korteweg--de Vries (KdV) equation
\begin{equation}\label{KdV}
u_t+uu'_x+u'''_{xxx}=0\end{equation}
and the modified Korteweg--de Vries (mKdV) equation
\begin{equation}\label{MKdV}
w_t-6w^2w'_x+w'''_{xxx}=0\end{equation}
respectively. The evolution equations (\ref{KdV}) and (\ref{MKdV}) are integrable by inverse
scattering transform (IST) and, by the results in~\cite{Dom}, all their local holomorphic solutions 
can be extended to the whole $x$-plane as globally meromorphic functions of the spatial variable.
In a similar vein, the global meromorphy of solutions of the initial-value problems was proved in
\cite{dobish} for the infinite hierarchies of higher analogues of the first and second Painlev\' e
ODE and reproved for the fourth Painlev\' e~ODE.

The main result of the present paper formally consists in demonstrating in sections %\ref{s3}--\ref{s5} 
III--V
that the global meromorphy of solutions of the initial-value problems can be proved in a uniform and
simple way for a very large class of non-linear ODE which determine the so-called
{\it invariant manifolds} of solutions of KdV (\ref{KdV}), mKdV (\ref{MKdV}), the complexified
non-linear Schr\" odinger (NLS) equation
\begin{equation}\label{NS}
p'_t=-ip''_{xx}+2ip^2q, \qquad q'_t=iq''_{xx}-2ipq^2
\end{equation}
and the Sawada--Kotera (SK) equation \cite{Sav}
\begin{equation}\label{SK}
u'_t=u'''''_{xxxxx}-30uu'''_{xxx}-30u'_{x}u''_{xx}+180 u'_xu^2. 
\end{equation}
As in \cite{dobish}, the main step in the proof of global meromorphy consists in using
the meromorphic extendibility with respect to~$x$ of all local holomorphic solutions of the evolution
equations (\ref{KdV}) -- (\ref{SK}). However, in contrast with \cite{dobish}, we consider
not only ODE for self-similar solutions of IST-integrable evolution equations but also 
ODE that determine the so-called {\it symmetries} (and, more generally, invariant manifolds)
of these evolution equations. Note that the ODE for self-similar solutions are very special
representatives of this class. They correspond to the so-called {\it classical} symmetries.

Our scheme for proving the meromorphy of solutions 
can be extended to many other simultaneous solutions
of ODE and IST-integrable evolution equations. The main and almost the only obstacle is the
necessity to establish that all local holomorphic solutions of the integrable partial differential 
equation admit a global meromorphic extension with respect to one of the independent variables.
However, this was already done in \cite{dobish}, \cite{Dom}--\cite{Domk} for many
equations including (\ref{KdV})--(\ref{SK}).

The main part of the paper begins in the next section, where we give the necessary auxiliary
information concerning some properties of systems of evolution equations integrable by 
inverse scattering transform.

%%%%%%%%%%%%%%%%

\section{\label{sec:level2}Symmetries and invariant manifolds of evolution equations}

Consider a system of evolution equations 
\begin{equation}\label{sise} 
(u^i)'_t=F^i(t,x,u^1,\dots, u^m_{n}), \qquad 1\leq i\leq m,
\end{equation}
where $F^i$, $1\leq i\leq m$, are locally analytic functions of
$t$, $x$, $u^m_{r}=\left(\frac{\partial}{\partial x}\right)^ru^m$.
(Here and in what follows, subscripts of
the components of solutions stand for the orders of their derivatives with respect to the 
spatial independent variable~$x$.) 

\begin{definition} \label{d1}
A symmetry of the system (\ref{sise}) is a system of evolution equations
\begin{equation}\label{sime} (u^i)'_{\tau}=G^i(t,x,u^1,\dots, u^m_{k}), \qquad 1\leq i\leq m, \end{equation}
with locally analytic right-hand sides $G^i$ satisfying the commutation condition
%\begin{equation}\begin{split}
%u_5=-\frac{5uu_3}{3}-\frac{10u_1u_2}{3}-\frac{5u^2u_1}{6}+\\
%\label{infnk}+\frac{1}{3t}\left[(u_3+uu_1)+4u_2+\frac{4u^2}{3}+\frac{u_1v}{3}\right] 
%\end{split}\end{equation}
%\begin{eqnarray}
\begin{equation}\frac{\partial G^i}{\partial t}
+\sum_{j=1}^m \sum_{l=1}^{k}\frac{\partial G^i}{\partial u^j_{l}}\left(\frac{d}{dx}\right)^lF^j= %\nonumber \\
%\left.\frac {d F^i}{d \tau}\right|_{(\ref{sime})}
\sum_{j=1}^m \sum_{l=1}^{n}\frac{\partial F^i}{\partial u^j_{l}}\left(\frac{d}{dx}\right)^lG^j. \nonumber
\end{equation}%\end{eqnarray}
\end{definition}

Note that $\tau$ does not occur in the right-hand side of~(\ref{sime}). 
Furthermore, it follows directly from the definition that all linear combinations 
(with constant coefficients~$a_j$) $\sum_{j=1}^ra_jG^i_j$,  $1\leq i\leq m$,
of the right-hand sides  $G^i_j$ of symmetries (\ref{sime}) of the system (\ref{sise}) are 
again right-hand sides of symmetries of this system.

Symmetries with functions of the form 
\begin{equation}\nonumber
G^i= W^i(t,x, u^1,\dots,u^m,u^1_1 ,\dots, u^m_1,F^1,\dots,F^m)
\end{equation}
(where $F^i=F^i(t,x,u^1,\dots, u^m_{n})$)
%%\begin{equation}
%%G^i=W^i(t,x, u^1,\dots,u^m,u^1_1 ,\dots, u^m_1) 
 %%, \dots, u^m_1,F^1(t,x,u^1,\dots, u^m_{n}),\dots,F^m(t,x,u^1,\dots, u^m_{n}))
%%\nonumber\end{equation}  
are referred to as classical symmetries of the system  (\ref{sise}) of evolution equations. 
All other functions $G^i$ determine {\it higher} symmetries of~(\ref{sise}). (The existence
of higher symmetries is characteristic for integrable systems of evolution equations.)
Systems of ODE given by the stationary parts $u^i_{\tau}=0$ of the symmetries~(\ref{sime})
are particular cases of systems of ODE that determine invariant manifolds of 
the system~(\ref{sise}). Note that self-similar solutions of partial differential equations
(integrable or not) are determined by the invariant manifolds corresponding
to classical symmetries of these equations.

\begin{definition} \label{d2}
We say that a system of ODE %%(with respect to~$x$)
 %\begin{equation}\label{invm}H^i(t,x,u^1,\dots, u^m_{k})=0  \quad 1\leq i\leq m\end{equation}
\begin{equation}\label{invm} 
u^i_{N_i} =S^i(t,x,u^1,\dots, u^m,\dots,u^1_{N_i-1},\dots,u^m_{N_i-1}),\qquad i=1,\dots m,
\end{equation} 
where the functions $S^i$ are analytic near the points
\begin{equation}\label{fix}
 (t,x,u^1,\dots, u^m,\dots,u^m_{N_i-1})=(t_0,x_0,a^1,\dots,a^m,\dots,a^m_{N_i-1})
\end{equation}
{\rm(}$t_0,x_0, a^i, a^i_{n_i}$ are arbitrary complex constants{\rm)}, determines an invariant manifold
of the system~(\ref{sise}) if the following relations hold on all solutions $u(t,x)$ of~(\ref{invm}) ($1\leq i\leq m$):
\begin{equation}\label{kritin}
 \frac {\partial S^i}{\partial t}+\sum_{j=1}^m \left[\frac {\partial S^i}{\partial u^j}F^j+\sum_{l=1}^{N_{i}-1}\frac {\partial S^i}{\partial u^j_l}\left(\frac{d}{dx}\right)^lF^j\right]=\left(\frac{d}{dx}\right)^{N_i}F^i. \end{equation}
\end{definition}

The relations (\ref{kritin}) are clearly necessary for the existence of a simultaneous local analytic 
solution of the systems (\ref{sise}) and (\ref{invm}). 

We now assume that %%$F^i$ and 
$S^i$ are local analytic functions near the points (\ref{fix}) 
and the relations (\ref{kritin}) hold. Consider the set of all analytic solutions of the system of ODE
(\ref{invm}) in a neighborhood of the point $(t_0,x_0) \in \mathbb{C}^2$ satisfying the conditions
\begin{equation}\label{kosh} 
u^i_{n_i}(t_0,x_0)=a^i_{n_i} \qquad 1\leq i \leq m, \qquad 0\leq n_i\leq N_i-1. 
\end{equation}
We claim that this set contains a unique solution of (\ref{sise}). Indeed, differentiating 
each component of the system of ODE (\ref{invm}) with respect to~$t$ and using
(\ref{kritin}), we see that the functions 
$$
w^i(x,t)=(u^i)'_t-F^i(t,x,u^1,\dots, u^m_{n}), \qquad 1\leq i\leq m, 
$$
satisfy the following system of linear ODE %%near the point (\ref{fix}):
in a neighborhood of the point $(t_0,x_0) \in \mathbb{C}^2$ ($1\leq i \leq m, \quad 0\leq n_i\leq N_i-1$):
\begin{equation}\label{diff} 
\left(\frac{d}{dx}\right)^{N_i}w^i=\sum_{j=1}^m \left[\frac {\partial S^i}{\partial u^j}w^j
+\sum_{l=1}^{N_{i}-1}\frac {\partial S^i}{\partial u^j_l}\left(\frac{d}{dx}\right)^lw^j \right].
\end{equation}
For every fixed $t$ in a neighborhood of~$t_0$, these ODE (in the independent variable~$x$)
are resolved with respect to the highest derivatives and their coefficients are analytic functions. 
By the uniqueness theorem, our task reduces to choosing the initial data
$$u^i_{n_i}(t,x_0)=v^{in_i}(t) \qquad (1\leq i\leq m, \quad 0\leq n_i \leq N_i-1)$$
of the considered solution of (\ref{invm}) in such a way that the functions $w^i(x,t)$
and all their derivatives in~$x$ of orders up to $N_i-1$ vanish identically in a neighborhood
of~$t_0$ for $x=x_0$.  But this can be done in a unique way, namely, letting $v^{in_i}(t)$
be the unique solution with initial conditions~(\ref{kosh}) of the system of first-order ODE
$$
(v^{in_i})'_t=G^{in_i}(t,x_0,v^{10},\dots,v^{m0},\dots,v^{1(N_1-1)},\dots,v^{m(N_m-1)}),
$$
where the right-hand side is defined as 
$\left.\left(\frac{d}{d x}\right)^{n_i}F^i(t,x,u^{1},\dots, u^{m}_n)\right|_{x=x_0}$
with all the occurring derivatives of $u^j$ of order greater than $N_j-1$ (if any) 
expressed in terms of the functions $v^{in_i}(t)$ by means of the system~(\ref{invm})
and its differential consequences. Thus we arrive at the following result.

\begin{theorem} \label{t1}
For any complex numbers $a^i_{n_i}$ $(1\leq i \leq m$, 
 $ 0\leq n_i\leq N_i-1)$ and all
sufficiently small $\delta_1,\delta_2>0$ there is a unique simultaneous analytic solution $u(t,x)$
in the bidisk
\begin{equation}\label{inip} 
\left\{ (x,t)\in\mathbb{C}^2\,|\,|t-t_0|<\delta_1,\, |x-x_0|<\delta_2 \right\}
\end{equation}
of the systems (\ref{sise}) and (\ref{invm}) with initial conditions~(\ref{kosh}).
\end{theorem}

\begin{remark} \label{r1}
{\rm This theorem was essentially proved in section 3.1 of \cite{Kap}, although in the infinitely
differentiable case instead of our analytic version. For the reader's convenience, we gave
a self-contained proof of Theorem~\ref{t1} which is simpler than the proof in~\cite{Kap}.}
 %In particular we do not use the classical Frobenius theorem.
\end{remark}

%%%%%%%%%%%%%%%%
\section{\label{sec:level3}Meromorphy of solutions of the stationary parts of symmetries
for KdV and mKdV}

We first explain our scheme of proving the meromorphy of solutions of ODE (arising as
invariant manifolds of evolution equations) on the example of the hierarchy of the first
Painlev\' e equation~(\ref{PI}). This hierarchy can be written in the form
\begin{equation}
 \label{IerPIo} 
x-tu +L_n(u,u_1,...u_{2n})+\sum_{j=1}^{n-1} \mu_j L_j(u,u_1,...u_{2j})=0, 
\end{equation}
where $\mu_j$ are arbitrary complex numbers and $L_j$ are the so-called Lenard polynomials,
which are defined recursively as

%\eqa{line1 \\ line2\\line3\\line4}
\begin{equation} \begin{split} \label{Ierlt} L_1=u_2+\frac{u^2}{2}, \quad 
L_2=u_4+\frac{5uu_2}{3}+\frac{5(u_1)^2}{6}+\frac{5(u)^3}{18}, \\
\frac{d}{dx} L_{m+1}=\left[\left (\frac{d}{dx} \right)^{3}+\frac{2u}{3}\frac{d}{dx}+
 \frac{u_1}{3} \right]L_m \end{split}\end{equation}
 and are identically equal to zero for $u\equiv 0$. 

This hierarchy of ODE is known \cite{Moo}, \cite{Kit} as the hierarchy of massive $(2n + 1, 2)$
string equations. It is also referred to as the $P_1^n$-hierarchy of the first Painlev\' e equation
(see, for example, \cite{Kit}, \cite{Nogen}).

Differentiating both parts of the ODE (\ref{IerPIo}) with respect to~$x$, we obtain the stationary
parts of symmetries of the KdV equation~(\ref{KdV}):
\begin{equation}\label{IerPIosim} 1-tu_1 +\frac{d}{dx}\left(L_n(u,u_1,...u_{2n})+\sum_{j=1}^{n-1} \mu_j L_j(u,u_1,...u_{2j})\right)=0. \end{equation}
These are linear combinations of the stationary part of the classical Galilean symmetry
 $u'_{\tau_G}=1-tu_1$ of the KdV equation and the stationary parts of representatives
of the well-known hierarchy of commuting higher symmetries
%\begin{eqnarray} 
\begin{equation} \begin{split} \label{com} u'_{\tau_1}=u_3+uu_1=K_3[u],\quad u'_{\tau_2}=\frac{d}{dx}(L_2)=K_5[u], \\ 
\quad \dots, \quad
 u'_{\tau_n}=\frac{d}{dx}(L_n)=K_{2n+1}[u], %\nonumber\end{eqnarray}
\end{split}\end{equation}
which can be defined in terms of the recursion operator 
\begin{equation} 
L_{KdV}=\frac{d^2}{dx^2}+\frac{2u}{3}+\frac{u_1}{3}\left(\frac{d}{dx}\right)^{-1}
\nonumber\end{equation}
by the formula $K_{2m+3}[u]=L_{KdV}K_{2m+1}[u]$ \cite[\S5.2, \S5.3]{Olv}. 

Since (\ref{IerPIosim}) is the stationary part of a symmetry of KdV, Theorem~1 yields that
for {\it arbitrary} complex constants $b_k$ and any point $(t=t_0$, $x=x_0)$ in the complex 
space $\mathbb{C}^2$ there is a unique simultaneous analytic (in $t$ and $x$)
solution of (\ref{KdV}) and (\ref{IerPIosim}) in a neighborhood (\ref{inip}) of this point
satisfying the conditions
\begin{equation} u_{k}(t_0,x_0)=b_k \qquad (0\leq k\leq 2n).
\label{Koshp}\end{equation}
By \cite{Dom}, every analytic solution of the KdV equation (\ref{KdV}) in the bidisk (\ref{inip})
 extends to a meromorphic function in the strip
\begin{equation} \label{Polos}
\left\{(x,t)\in\mathbb{C}^2\,|\,|t-t_0|<\delta_1\right\}.
\end{equation} 
Hence, for any fixed complex values of $t_0$ and $b_k$ $(0\leq k\leq2n)$,
the unique solution of the ODE (\ref{IerPIosim}) with $t=t_0$ satisfying the initial conditions
(\ref{Koshp}) for $x=x_0$ extends to a meromorphic function of $x,t$ in the strip
(\ref{Polos}).  Since the constants $b_k$ are arbitrary, we obtain the global meromorphy
with respect to~$x$ of all solutions of the ODE
$$
x-t_0u +L_n(u,u_1,...u_{2n})+\sum_{j=1}^{n-1} \mu_j L_j(u,u_1,...u_{2j})=0
$$
with arbitrary initial data 
\begin{equation} 
u_{k}(t_0,x_0)=c_k \qquad (0\leq k \leq 2n-1). 
\label{fKVn}\end{equation}
This proves the following theorem.

\begin{theorem}\label{t2}
For any values of the complex constants $x_0$, $t_0$ , $\mu_1$, $\dots,$ $\mu_{2n}\neq 0$ and $c_0$, $\dots$, $c_{2n-1}$, the solution of the Cauchy problem for the equation
 (\ref{IerPIo}) with initial data (\ref{fKVn})  is meromorphic with respect to $x$ 
in the whole complex plane~$\mathbb{C}$.  %% not in a strip (\ref{Polos}).
\end{theorem}

Theorem~\ref{t2} was already proved in our previous paper~\cite{dobish}, but the present
proof is considerably simpler (for $n > 1$) since there is no need to represent the ODE 
(\ref{IerPIo}) as reductions of higher analogues of the KdV equation and give a separate
proof of meromorphic extendability of solutions of these analogues.

\begin{remark}\label{r2}
{\rm
Besides the ODE $P_1^1$, which is clearly equivalent to the first Painlev\' e equation
(\ref{PI}), the only rigorous and detailed proof of the meromorphy of solutions of~(\ref{IerPIo}) 
prior to \cite{dobish} was given for $P_1^2$ in \cite{Shimt}. We note that
 Shimomura~\cite{Shimn} announced a proof of meromorphy of solutions of  (\ref{IerPIo}) 
for all positive integers~$n$ in~2004, but the subsequent authors
 \cite{Grom}--\cite{Ste} still regard the question of rigorous proof for $n>2$ as open.
}
\end{remark}

\begin{remark}\label{r3}
{\rm
Simultaneous solutions of the ODE (\ref{IerPIo}) for $n=2$ and the KdV equation (\ref{KdV})
are equivalent \cite{Shimt} to solutions of a pair of isomonodromic Hamiltonian systems
$H^{9/2}$ with two degrees of freedom belonging to a hierarchy (written out by Kimura~\cite{Kim})
of degenerations of the classical two-dimensional Garnier system.
}
\end{remark}

\begin{remark}\label{r4}
{\rm
By \cite{Bis}, the set of these simultaneous solutions contains 
the familiar Gurevich--Pitaevskii solution
(introduced in \cite{Gur}, \cite{Pit}) of the KdV equation. This special solution describes the
generation of the so-called dispersive shock waves for a wide class of problems with small
dispersion in the leading order. The same special solution of~(\ref{IerPIo})  was considered
in the early nineties of the last century (see \cite{Bis}, \cite{Kupl}, \cite{Gapl}) in connection
with the problems of quantum gravitation theory (\cite{Moo}, \cite{Bre}, \cite{Sei}) and the
description of blow-up regimes in the random matrix theory (\cite{Its}, \cite{Van}).
Many publications in the last quarter century were devoted to studying various properties
of simultaneous solutions of ODE and KdV %%(\ref{KdV}) 
(first of all, the Gurevich--Pitaevskii special solution); see \cite{Kud}--\cite{Shav} and
references therein in addition to those given above. 
Other representatives of the hierarchy of simultaneous solutions of
 (\ref{IerPIo}) and (\ref{KdV}) are also related to the description of formation of dispersive 
shock waves in degenerate cases \cite{Nogen}, \cite{Kud}.
}
\end{remark}

Our scheme can be used in an equally simple way to prove the meromorphy in~$x$
of all solutions of initial-value problems for the following hierarchy of stationary parts 
of symmetries of~KdV (\ref{KdV}):
\begin{equation}\label{IerPIIksim} 
2u+xu_1-3t(u_3+uu_1) +\frac{d}{dx}\left(\sum_{j=1}^{n} \mu_j L_j(u,u_1,...u_{2j})\right)=0, \end{equation}
where $\mu_j$ are arbitrary complex constants. Here the stationary part
of the classical scaling symmetry $u'_{\tau_r}= 2u+xu_1-3t(u_3+uu_1)$ \cite[\S5.2]{Olv} 
is added to a linear combination of the stationary parts of the symmetries (\ref{com}). 
Being the stationary part of a symmetry of~KdV, the ODE~(\ref{IerPIIksim}) determines
an invariant manifold of~KdV. Assuming that $\mu_n\neq0$, we see from Theorem \ref{t1}
that every initial condition~(\ref{Koshp})
 % \begin{equation} u_{k}(x_0)=c_k, \qquad 0\leq k \leq 2n, \label{nachk}\end{equation}
determines a unique simultaneous holomorphic solution of the ODE (\ref{IerPIIksim}) and 
the KdV equation~(\ref{KdV}) in a neighbourhood (\ref{inip}) of an arbitrary point $(t_0,x_0)$.
By~\cite{Dom}, any such solution of~(\ref{KdV}) extends meromorphically to the 
strip~(\ref{Polos}). Hence all solutions of (\ref{IerPIIksim}) are meromorphic 
with respect to $x,t$ in this strip and, in particular, meromorphic on the whole $x$-plane 
for every fixed~$t$. This proves the following theorem.

\begin{theorem}\label{t3}
Let $\mu_1,\dots,\mu_n, x_0, t_0$ and $b_0,\dots,b_{2n}$ be any complex constants.
Suppose that $n>1$ and $\mu_n\neq0$. Then the unique solution of the ODE (\ref{IerPIIksim}) 
with $t=t_0$ and
initial conditions (\ref{Koshp}) is meromorphic with respect to~$x$ on the whole complex 
plane~$\mathbb{C}$.  %% not in the strip (\ref{Polos}).
\end{theorem}

The lowest representative of the hierarchy of ODE (\ref{IerPIIksim}) corresponds to the case
when $\mu_1=\dots=\mu_n=0$. It coincides with the stationary part of the classical scaling
symmetry of the KdV equation (\ref{KdV}),
\begin{equation} 2u+xu_1-3t(u_3+uu_1)=0. \label{scail}\end{equation}
The proof of Theorem~\ref{t3} shows that if $t \neq 0$, then any solution of this equation
(that is, the solution with initial conditions $u_k|_{x=x_0}= c_k$, $0\leq k\leq2$, for arbitrary 
complex constants $x_0,c_0,c_1,c_2$) is also meromorphic on the whole plane. The change
of variables
$$
z=\frac{x}{(-3t)^{1/3}},\qquad u=(-3t)^{2/3}v(z)
$$
reduces this ODE as well as the KdV equation~(\ref{KdV}) to the $t$-independent ODE
$$ 
v'''_{zzz} + (v+z)v'_z + 2v= 0.
$$
The familiar Miura transform
$v=6(f'_z-f^2)$ sends the solutions of this ODE to solutions of the second Painlev\' e
equation~(\ref{PII}), which has many applications to the problems of non-linear mathematical
physics.

\begin{remark}\label{r5}
{\rm Simultaneous solutions of the KdV equation (\ref{KdV}) and the representative of the hierarchy
(\ref{IerPIIksim}) with $n=2$ and $\mu_2\neq0$ determine simultaneous solutions (see
\cite{Pav}) of a pair of isomonodromic Hamiltonian systems $H^{7/2+1}$ with two degrees of
freedom belonging to Kawamuko's list~\cite{Kawam}. This list of Hamiltonian systems also
consists of degenerations of the classical isomonodromic Garnier pair, but they are different 
from the hierarchy of its degenerations written out in~\cite{Kim}. 
}
\end{remark}

\begin{remark}\label{r6}
{\rm Special simultaneous solutions of  (\ref{KdV}) and the representative of (\ref{IerPIIksim}) 
with $n=2$ and $\mu_2\neq0$ were studied in \cite{Garsu}--\cite{Garut}. They give a
universal description~\cite{Garsu} of the corrective influence of small dispersion on the process
of transforming weak discontinuities of solutions of ideal hydrodynamic equations into
strong discontinuities. It is also clear that solutions of other higher representatives of
the hierarchy of ODE (\ref{IerPIIksim}) also describe this influence in certain degenerate 
cases, for example, in the case considered in \cite[\S3.2]{Kamch}.
}
\end{remark}

%The recursion operator sends the right-hand side $u'_{\tau_G}=1-tu_1$  of the classical Galilean
%symmetry of the KdV equation \ref{KdV}) to the right-hand side
%$u'_{\tau_r}= 2u+xu_1-3t(u_3+uu_1)$ of its classical scaling symmetry.

Note that the proof in  \cite{dobish} of the global meromorphy property for all solutions of 
the higher analogues of the second Painlev\' e equation~(\ref{PII}) can also be simplified. 
Indeed, differentiating these analogues with respect to~$x$, we have the following
hierarchy of  ODE ($n=2,3,\dots $):
\begin{equation}\label{IePtwoo} 
\frac{d}{dx} \left(t(2w^3-w_2)+xw+\sum_{j=2}^{n}{\nu_j}P_{j}(w,w_1,\dots,w_{2j})\right)=0,
\end{equation}
where $\nu_j$ are arbitrary complex numbers, $\nu_n\neq 0$. These ODE are obtained by
adding the stationary part of the classical scaling symmetry
$w'_{\tau}=3t(w_3-6w^2w_1)+xw_x+w$ of the mKdV equation (\ref{MKdV}) 
to a linear combination of the stationary parts of higher symmetries of %%(\ref{MKdV}) 
this equation (see, for example, \cite{Cl}),
$$ 
w'_{\tau_n}=\frac{d}{dx}P_{n}(w,w_1,\dots,w_{2n})=
\frac{d}{dx}\left(\frac{d}{dx}+w\right)L_n[w_x-w^2],
$$
where $L_n[u]=L_n(u,u_1,...u_{2n})$ are the Lenard polynomials (\ref{Ierlt}). Using
Theorem~\ref{t1} and the result of~\cite{Dom} about the global meromorphic extensibility with
respect to~$x$ of all local holomorphic solutions of the mKdV equation, we see that for
arbitrary complex parameters $t=t_0$, $x_0$ and $a_i$ $(i=0,1,\dots,2n)$ the solution 
of~(\ref{IePtwoo}) with initial conditions $u_{i}(x_0)=a_i$ $(i=0,1,\dots,2n)$ is meromorphic
on the whole complex $x$-plane. Hence, for arbitrary complex values of %the complex constants
$\alpha_n$, $t=t_0$, $x_0$ and $a_i$ $(i=0,1,\dots,2n-1)$, the solution of the $n$\,th
equation of the second Painlev\'e hierarchy
\begin{equation}\label{Ptooi} 
t(w_2-2w^3)+xw+\sum_{j=1}^n\nu_j P_{j}(w,w_1,\dots,w_{2j})=\alpha_n
\end{equation}
with initial conditions $u_{i}(x_0)=a_i$ $(i=0,1,\dots,2n-1)$ is meromorphic on the
whole complex $x$-plane. Moreover, any local holomorphic solution of (\ref{Ptooi}) is 
meromorphic with respect to $x,t$ in some strip~(\ref{Polos}).

\begin{remark}\label{r7}
{\rm The hierarchy of the second Painlev\' e equation and its relation to the hierarchy of ODE
(\ref{IerPIIksim}) was earlier described in \cite{ClJ} --\cite{Mo}.
} 
\end{remark}

Our scheme can also be used to prove the global meromorphy of solutions of ODE given by
the stationary parts of {\it non-local} symmetries of the KdV equation~(\ref{KdV}). (The right-hand
side of a non-local symmetry can depend not only on $t$, $x$, $u$, $u_1$, $\dots$, $u_m$
but also on the integrals of polynomials in the function $u$ and its derivatives with respect to~$x$.)
For example, consider the simplest non-local symmetry of~(\ref{KdV}), the so-called master
symmetry
%\begin{equation}\begin{split} 
%\eqa{line1 \\ line2\\line3\\line4}
\eqa{\label{mask}
u'_{\tau_{mas}}=-3t\left[u_5+\frac{5uu_3}{3}+\frac{10u_1u_2}{3}+\frac{5u^2u_1}{6}\right]+\\+x(u_3+uu_1)+4u_2+\frac{4u^2}{3}+\frac{u_1v}{3}=G^1(t,x,u,\dots,u_5,v),} 
%\end{split}\end{equation}
where
\begin{eqnarray}
%\eqa{
\label{npk}
v_x&=&u,\\
\label{ntk}v_t&=&-u_2-\frac{u^2}{2}.\end{eqnarray}
The right-hand side $G^1(t,x,u,\dots,u_5,v)$ of this symmetry, which was introduced
in the end of section~2 of~\cite{Ibsh}, satisfies the following 
identity %analogue of Definition\,\ref{Defd1}:
\begin{equation}\begin{split} 
\frac {\partial G^1}{\partial t}-(u_3+uu_1)\frac {\partial G^1}{\partial u}-\left(u_2+\frac{u^2}{2}\right)\frac {\partial G^1}{\partial v}-\\ \label{uskok}-\sum_{j=1}^5\frac {\partial G^1}{\partial u_j}\left(\frac{d}{dx}\right)^j(u_3+uu_1)= 
-u_1G^1-u\frac {d G^1}{d x}-\frac {d^3 G^1}{d x^3}. 
\end{split}\end{equation} 
The stationary part $G^1(t,x,u,\dots,u_5,v)=0$ of the master symmetry (\ref{mask}),
after dividing both parts by $(-3t)$, takes the form 
\begin{equation}\label{stank}g^1(t,x,u,\dots,u_5,v)=0,\end{equation}
where the function
\begin{eqnarray}
 &g^1(t,x,u,\dots,u_5,v)=u_5+\frac{5uu_3}{3}+\frac{10u_1u_2}{3}+\frac{5u^2u_1}{6}-
\nonumber \\
&
-\frac{1}{3t}\left\{x(u_3+uu_1)+4u_2+\frac{4u^2}{3}+\frac{u_1v}{3}\right\}
\nonumber\end{eqnarray}
is not the right-hand side of a symmetry of KdV. Nevertheless, by
(\ref{uskok}), the following equality holds on solutions of the ODE (\ref{stank}):
\begin{eqnarray}
\frac {\partial g^1}{\partial t}-(u_3+uu_1)\frac {\partial g^1}{\partial u}-\left(u_2+\frac{u^2}{2}\right)\frac {\partial g^1}{\partial v}-\nonumber\\ -\sum_{j=1}^5\frac {\partial g^1}{\partial u_j}\left(\frac{d}{dx}\right)^j(u_3+uu_1) 
=-\left(u_1g^1+u\frac {d g^1}{d x}+\frac {d^3 g^1}{d x^3}\right). \nonumber
\end{eqnarray}
This equality means that the function
\begin{eqnarray}
S^2(t,x,u,u_1,\dots,u_4,v)=-\frac{5uu_3}{3}-\frac{10u_1u_2}{3}-\frac{5u^2u_1}{6}+\nonumber\\ +
\frac{1}{3t}\left\{x(u_3+uu_1)+4u_2+\frac{4u^2}{3}+\frac{u_1v}{3}\right\} \nonumber
\end{eqnarray}
satisfies the following condition:
\begin{eqnarray}
\frac {\partial S^2}{\partial t}-\frac {\partial S^2}{\partial u}(u_3+uu_1) -\frac {\partial S^2}{\partial v}(u_2+u^2/2)+
\nonumber \\+\sum_{l=1}^{4}\left[\frac {\partial S^2}{\partial u_l}\left(\frac{d}{dx}\right)^l(u_3+uu_1)\right] 
=-\left(\frac{d}{dx}\right)^{5}(u_3+uu_1). \nonumber
\end{eqnarray}
Combining this with obvious identity
$$
\frac {\partial S^1}{\partial t}-\frac {\partial S^1}{\partial u}(u_3+uu_1) =
-\frac{d}{dx}(u_2+\frac{u^2}{2})
$$
for the function $S^1=u$ and using Definition 2, we see that the system of ODE consisting
of (\ref{npk}) and the equation
\begin{equation}\begin{split}
u_5=-\frac{5uu_3}{3}-\frac{10u_1u_2}{3}-\frac{5u^2u_1}{6}+\\
\label{infnk}+\frac{1}{3t}\left[(u_3+uu_1)+4u_2+\frac{4u^2}{3}+\frac{u_1v}{3}\right] 
\end{split}\end{equation}
determines an invariant manifold of the system of evolution equations
(\ref{ntk}) and (\ref{KdV}). Hence, by Theorem\,\ref{t1}, for any six complex constants
$a$ and $b_i$ $(i=0,\dots, 4)$ there is a unique simultaneous  analytic solution of the system
of evolution equations (\ref{ntk}), (\ref{KdV})  and the system of ODE (\ref{npk}), (\ref{infnk})
in a neighborhood  (\ref{inip}) of any fixed point $(t_0,x_0)$ with initial conditions
\begin{equation}\label{inisn}
v(t_0,x_0)=a,\qquad u_i(t_0,x_0)=b_i \quad (i=0,\dots, 4).
\end{equation}
Using the uniqueness of this solution, the result of \cite{Dom} on the meromorphic
extension of its component~$u$ (a solution of the KdV equation (\ref{KdV})) to the strip 
(\ref{Polos}), and the shape of the ODE (\ref{infnk}), we obtain the following theorem.

\begin{theorem}\label{t4}
For any values of the complex constants  $x_0$, $t_0\neq 0$, $a$ and $b_i$ $(i=0,\dots, 4)$,
the solution of the system of ODE (\ref{npk}), (\ref{infnk}) with $t=t_0$ and 
initial conditions (\ref{inisn}) is meromorphic with respect to $x$ on the whole complex
plane~$\mathbb{C}$. %%in some strip~(\ref{Polos}) disjoint from the complex line $t=0$.
\end{theorem}

An analogous theorem can be proved by the same scheme for the system of ODE
consisting of (\ref{npk}) and the stationary part
\begin{equation}\label{masskl}G^2(t,x,u,\dots,u_5,v)=0\end{equation}
of the following symmetry of the KdV equation (\ref{KdV}):
\begin{eqnarray}
u'_{\tau_{mascl}}=G^2(t,x,u,\dots,u_5,v)=G^1(t,x,u,\dots,u_5,v) +\nonumber\\
+k_1(3tu_3+xu_1+2u)+k_2(1-tu_1).
\nonumber\end{eqnarray}
This symmetry generalizes the master symmetry (\ref{mask}) to the case of arbitrary
complex constants $k_1$ and $k_2$.

\begin{theorem}\label{t5}
Let $t_0\neq 0$ and $k_1,k_2$ be arbitrary complex para\-meters. Then, for any values of the
complex constants  $x_0$, $a$ and $b_i$ $(i=0,\dots, 4)$, the solution of the system of ODE
(\ref{npk}), (\ref{masskl}) with $t=t_0$ and initial conditions~(\ref{inisn}) is meromorphic 
with respect to~$x$ on the whole complex plane~$\mathbb{C}$. %%in some strip~(\ref{Polos}).
\end{theorem}

\begin{remark}\label{r8}
{\rm
Recently Adler \cite{Adler} described a simultaneous solution of the system of evolution equations
(\ref{ntk}),  (\ref{KdV}) and the system of ODE (\ref{npk}), (\ref{masskl}). This solution provides
 an example
of an exact solution of the so-called first Gurevich--Pitaevskii problem \cite[\S8]{NovD} 
concerning solutions
of the KdV equation (\ref{KdV}) with step-like initial data possessing different constant limits as
 $x\to \pm \infty$. 
}
\end{remark}

Repeating the proofs of Theorems \ref{t2} and \ref{t5}, we arrive at the following result 
concerning the whole hierarchy of higher analogues of the system of ODE (\ref{npk}) and
\begin{eqnarray}
 G^2(t,x,u,\dots,u_5,v)+ \frac{d}{dx}\left(L_n(u,u_1,...u_{2n})+\right.\nonumber \\ \label{massh} \\ 
\left.+\sum_{j=1}^{n-1} \mu_j L_j(u,u_1,...,u_{2j})\right)=0, \nonumber\end{eqnarray}
where $L_j(u,u_1,...u_{2j})$ are the Lenard polynomials (\ref{Ierlt}).

\begin{theorem}\label{t6}
Let $t_0$, $k_1,k_2$ and $\mu_j$ be arbitrary complex para\-meters. Then, 
for any values of the
complex constants  $x_0$, $a$ and $b_i$, the solution of the system of ODE 
(\ref{npk}), (\ref{massh}) with $t=t_0$ and $n>2$ satisfying the initial conditions
$$ 
v(t_0,x_0)=a,\qquad u_i(t_0,x_0)=b_i \quad (i=0,\dots, 2n)
$$
is meromorphic with respect to~$x$ on the whole complex plane~$\mathbb{C}$. 
\end{theorem}

\section{\label{sec:level4}Meromorphy of solutions of the stationary parts of symmetries for the complexified
non-linear Schr\" odinger equation} 

Along with the hierarchy (\ref{Ptooi}) of the second Painlev\' e equation~(\ref{PII})
(this hierarchy is related to symmetries of the mKdV equation (\ref{MKdV})), there is
another its hierarchy, which seems to be first described in~\cite{Kit}. The $n$\,th
element of this hierarchy is a system of two ODE of the form
\begin{eqnarray}
\sum_{j=2}^n\mu_jg_j[p,q]-2tp_1-ixp=0, \nonumber \\ \label{IeIIPII} \\ \sum_{j=2}^n\mu_j\psi_j[p,q]-2tq_1+ixq=0,
\nonumber \end{eqnarray}
where  $\mu_j$ are arbitrary complex constants, $\mu_n\neq 0$, $n>1$. It is obtained
by adding the stationary part of the classical Galilean symmetry 
$$
p'_{\tau_G}=-2tp_1-ixp, \qquad q'_{\tau_G}=-2tq_1+ixq
$$
of the complexified NLS equation (\ref{NS}) to a linear combination of the stationary parts
of autonomous symmetries of~(\ref{NS}),
\begin{equation}\label{symNS}
p'_{\tau_n}=g_j[p,q], \qquad q'_{\tau_G}=\psi_j[p,q],\end{equation}
whose right-hand sides
%\begin{eqnarray}
\eqa{\label{AKNSh}
g_0[p,q]=ip, \quad \psi_0[p,q]=-iq, \\
g_1[p,q]=p_1, \quad \psi_1[p,q]=q_1, \\  
g_2[p,q]=-ip_2+2ip^2q, \quad \psi_2[p,q]=iq_2-2ipq^2, \\
g_3[p,q]=-p_3+6pqq_1, \quad \psi_2[p,q]=-q_3+6pqp_1, \quad \dots,}%\end{eqnarray}
are related \cite{Zh} by the recursive formula
\begin{equation} \begin{split}\label{Zhib}
 g_{j+1}=-i\frac{d}{dx}g_j+2ip\left(\frac{d}{dx}\right)^{-1}(p\psi_j+qg_j),
\\ \psi_{j+1}=i\frac{d}{dx}\psi_j-2iq\left(\frac{d}{dx}\right)^{-1}(p\psi_j+qg_j).
\end{split} \end{equation}
Here the integrals $\left(\frac{d}{dx}\right)^{-1}(p\psi_j+qg_j)$ are uniquely determined by 
requiring that they are equal to zero for $p=q=0$. It is easily provable by induction that the functions
$g_j[p,q]$ and $\psi_j[p,q]$ are polynomials of the form
 
\begin{equation} \begin{split}\nonumber g_j[p,q]=A_jp_j+F_j(p_{j-2}, q_{j-2},\dots, p,q),\\
\psi_j[p,q]=B_jq_j+G_j(p_{j-2}, q_{j-2},\dots, p,q)
\end{split} \end{equation}
for some non-zero constants $A_j,B_j$ such that the expressions
$p\psi_j+qg_j$ and $p\frac{d\psi_j}{dx}-q\frac{dg_j}{dx}$ can be written as the derivatives
\begin{equation}\begin{split}\nonumber
 p\psi_j+qg_j=\frac{d}{dx}\left(P_j(p_{j-1},q_{j-1},\dots,p,q)\right), \\ 
p\frac{d\psi_j}{dx}-q\frac{dg_j}{dx}=\frac{d}{dx}\left(Q_j(p_{j},q_{j},\dots,p,q)\right)\end{split} \end{equation}
of some polynomials $P_j(p_{j-1},q_{j-1},\dots,p,q)$ and $Q_j(p_j,q_j,\dots,p,q)$.

By shifting the independent variables $x,t$, the lowest representative of this hierarchy 
of systems of ODE can be written without loss of generality in the form
%\begin{equation}
\eqa{\label{sisp}
-i\mu_2[p_2-2p^2q]-2tp_1-ixp=0, \\ i\mu_2[q_2-2pq^2]-2tq_1+ixq=0.}
%\end{equation}
The familiar change of variables
\begin{equation}\label{chdo}
k=pq,\qquad l=\frac{q'_x}{q}
\end{equation}
 (see, for example, the formula (1.5) in~\cite{dobish})
transforms the complexified NLS equation (\ref{NS}) to the system of long water waves with
purely imaginary time,
\begin{equation}\label{DDV}
k'_t=i(2k'_1 l+2kl'_1-k_{2}),\qquad l'_t=i(l_{2}+2ll'_1-2k'_1).
\end{equation}
(The change of variables (\ref{chdo}) actually sends the solutions of the last system
to solutions of a somewhat more general system 
$$
p'_t=-ip_{2}+2ip^2q +b(t)p, \qquad
q'_t=iq_{2}-2ipq^2-b(t)q.
$$
But the simple transformation
$$
p(t)=\exp\left(\int_{t_*}^tb(\tau)\,d\tau\right)P, \qquad q(t)=\exp\left(-\int_{t_*}^tb(\tau)\,d\tau\right)Q
$$
again reduces them to solutions $(P,Q)$ of the complexified NLS equation (\ref{NS}).) 

The change of variables (\ref{chdo}) transforms the system of ODE (\ref{sisp}) 
to the following system of two first-order ODE depending on a complex parameter~$\gamma$:
\begin{equation}\label{ddvt}
k_1=2kl+\frac{2itk}{\mu_2}-\gamma \qquad l_1=-l^2-\frac{2itl}{\mu_2}+2k+\frac{ix}{\mu_2}.
\end{equation}
Indeed, it follows from (\ref{sisp}) that
$$ 
(p'_1q-q'_1p)'_x=\frac{2itk_1}{\mu_2}
$$
or, after integration,
$$
p'_1q-q'_1p=\frac{2itk}{\mu_2}-\gamma. 
$$
Here $\gamma$ is independent of $t$ since $(p,q)$ is a solution of the complexified
NLS equation~(\ref{NS}). Hence the functions (\ref{chdo}) constructed from any solution 
of~(\ref{sisp}) satisfy the first equation in~(\ref{ddvt}).  A direct verification shows that they
also satisfy the second equation in~(\ref{ddvt}). In its turn, the system~(\ref{ddvt}) 
is related by the formulas
$$ 
k=-v, \qquad l=-u-\frac{it}{\mu_2}
$$
with the equivalent system of ODE
$$
u_1=u^2+2v-\frac{ix}{\mu_2}+\frac{t^2}{\mu_2^2},\qquad v'_1=-2uv+\gamma, 
$$
which coincides with the system of ODE (69), (70) in \cite[\S5.1]{CJP}. The component~$u$
of a solution of this system satisfies the ODE
$$
u_{2}=2u^3+2u\left(\frac{t^2}{\mu_2^2}-\frac{ix}{mu_2}\right)-\frac{i}{\mu_2}+2\gamma,
$$
which is related by the trivial transformations 
$$
 w=\left(\frac{i\mu_2}{2}\right)^{1/3}u, \qquad z=\left(\frac{i\mu_2}{2}\right)^{-1/3}\left(x+
\frac{it^2}{\mu_2}\right)
$$
to the second Painlev\' e equation~(\ref{PII}). Hence the hierarchy of systems of ODE (\ref{IeIIPII})
may indeed be regarded as a hierarchy of higher analogues of~(\ref{PII}). (Its solutions are possibly
related by the change of variables (\ref{chdo}) to solutions of another second Painlev\' e  hierarchy
introduced in~\cite{CJP}.)

%Преобразования 
%\begin{equation}\nonumber p=u\exp\left(\frac{itx}{\mu_2}\right), \qquad q=v\exp\left(\frac{-itx}{\mu_2}\right)\end{equation}
%c последующим переходом к независимой переменной $y=-(\mu_2)^{-1/3}\left(x-\frac{t^2}{\mu_2}\right)$ перводят (\ref{sisp}) в систему
%ОДУ
%\begin{equation}\label{compP} u''_{yy}=2u^2v+yu=0, \qquad v''_{yy}=2uv^2+yv=0 \end{equation}
%Общее же решение последней формулами 
%\begin{equation}\nonumber w(z)=-2^{-1/3}\frac{v'_y(y)}{v(y)},\qquad z=-2^{1/3}y, \qquad \alpha=\frac{u'_y(y)v(y)-v'_y(y)u(y)}{4}-\frac{1}{2},\end{equation} 
%выведенной в (??) работе \cite{Jim}, связаны со вторым ОДУ Пенлеве (\ref{PII}). (Именно этот факт дает основание для названиия серии ОДУ (\ref{IeIIPII}) еще одной иерархией второго уравнения Пенлеве.) 

Being the stationary part of a symmetry of (\ref{NS}), the system of ODE (\ref{IeIIPII}) 
determines an invariant manifold of~(\ref{NS}). By Theorem \ref{t1}, for any complex constants
$a_{00},..., a_{(2n-1)0}$ and $b_{00},..., b_{(2n-1)0}$ there is a unique simultaneous 
holomorphic solution of the system of ODE (\ref{IeIIPII}) and the evolutionary system (\ref{NS})
in a neighborhood  (\ref{inip}) of an arbitrary point $(t_0,x_0)\in\mathbb{C}^2$ with initial conditions
\begin{equation} \label{inipNS}
p_j(t_0,x_0)=a_{j0},\qquad q_j(t_0,x_0)=b_{j0} \qquad (j=0,\dots,n-1).
\end{equation}
By~\cite{Dom}, every holomorphic solution of (\ref{NS}) in the bidisk (\ref{inip}) extends 
meromorphically in $x,t$ to the strip~$|t-t_0|<\delta_1$. Hence the following theorem holds.

\begin{theorem}\label{t7}
For any complex constants $t_0$, $x_0$, $a_{00}$,$...$,$a_{(2n-1)0}$ and 
$b_{00}$,$...,$ $b_{(2n-1)0}$, the solution $(p,q)$ of the system of ODE (\ref{IeIIPII}) with
$t=t_0$ and 
initial conditions (\ref{inipNS}) is meromorphic with respect to~$x$ on the whole complex
plane~$\mathbb{C}$. %% полосе (\ref{Polos}).
\end{theorem}

\begin{remark}\label{r9}
{\rm After the reduction $p=\delta q^*$ with real $\delta$, some solutions of the ODE
$\beta(q_3-6\delta|q|^2q_1)-2tq_1-ixq=0$ in this hierarchy coincide \cite{Bisl}, \cite{Kitj} with
the Haberman--Sun special solution~\cite{Habs} of the NLS equation $-iq_t=q_2-2\delta|q|^2q$.
This solution tends to the Pearcey integral
$const\int_R\exp[-i(\beta\lambda^4+4t\lambda^2+2x\lambda]\,d\lambda$ as $\delta \to0$
and describes the influence of small non-linearities on the high-frequency asymptotics near
the cusps of caustics. Other solutions of this ODE with $\delta=-1$ give a
prin\-ci\-pal-or\-der description \cite{Bisk} of the influence of small dispersion 
on the processes of falling
self-steepening of pulses, which are typical for non-linear geometric optic approximations.
By \cite{kaubb}, special solutions of other representatives of the hierarchy of ODE (\ref{IeIIPII})
give a hierarchy of solutions of the NLS equation $-iq_t=q_2-2\delta|q|^2q$ which
describe the influence of small non-linearities on the high-frequency asymptotics under
surgeries of caustics.
} 
\end{remark}

Our scheme gives an equally simple way to prove the global meromorphy in~$x$
of the solutions of initial-value problems for the hierarchy of systems of ODE
%\begin{equation}
\eqa{\label{IerPIV}
\sum_{j=0}^n\mu_jg_j[p,q]+2it(p_2-2p^2q) -xp_1-p=0, \\ \sum_{j=2}^n\mu_j\psi_j[p,q]-2it[q_2-2pq^2]-xq_x-q=0,}
%\end{equation}
where $\mu_j$ are any complex constants and $\mu_n\neq 0$. These systems of ODE are obtained
by adding the stationary part of the classical scaling symmetry  
%\begin{equation}
\eqa{\label{GNS}
p'_{\tau_{sc}}=2it(p_2-2p^2q) -xp_1-p, \\ q'_{\tau_{sc}}=-2it(q_2-2pq^2)-xq_x-q}
%\end{equation}
of the complexified NLS equation (\ref{NS}) to a linear combination of the stationary parts
of the autonomous symmetries~(\ref{symNS}). 

\begin{remark}\label{r10}
{\rm As already mentioned in \cite{Bvini}, it seems that the set of solutions of
(\ref{IerPIV}) with $n=3$ contains the special solution of the NLS equation
$-iq_t=q_2-2\delta|q|^2q$ described in~\cite{Sunh}. (When $\delta = 0$, this solution reduces
to $\alpha\int_R\exp[i(\beta\lambda^3-t\lambda^2+x\lambda]\,d\lambda$, where $\alpha$
and $\beta$ are real constants.)
}
\end{remark}

In view of the possibility of shifts in~$x$, there is no loss of generality in writing
the lowest representative of the hierarchy of ODE (\ref{IerPIV}) in the form of the system
\eqa{\nonumber
2it[p_2-2p^2q]-xp_1-p(1-i\mu_0)=0, \\ -2it[q_2-2pq^2]-xq_1-q(1+i\mu_0)=0.}
The change of variables (\ref{chdo}) transforms simultaneous solutions of the complexified
NLS equation (\ref{NS}) and this system to simultaneous solutions of the evolutionary system
(\ref{DDV}) and the system of ODE
%\begin{equation}
\eqa{\label{ddch} 
2it[2k_1l+2k'_1-k_{2}]=-2k-xk'_1, \\ 2it[l_{2}+2ll'_1-2k'_{1}]=-l-xl_1.}
%\end{equation}
For $t\neq 0$, this system can be reduced to the fourth Painlev\' e equation; see the end of
section~1 in~\cite{dobish}. (Integrating the last equation in (\ref{ddch}) with respect to $x$, 
we have
\begin{equation}\label{conkl}
k =\frac{l_1+l^2}{2}+\frac{xl-a}{4it},
\end{equation}
where the integration constant $a$ is independent of~$t$ since the functions $k(t,x)$ and $l(t,x)$
satisfy the system~(\ref{DDV}). Substituting the  expression (\ref{conkl}) for~$k$ into the first
equation in (\ref{ddch}) and making the changes of variables
$$
z=\frac {(i)^{1/2} x}{2t^{1/2}}, \qquad w=\frac {(i)^{1/2} }{2t^{1/2}}\left[l-\frac{ix}{2t}\right],
$$
we arrive at the ODE
\begin{equation}\label{slch}
w'''_{zzz}=6w^2w'_z+12zww'_z+4w^2+4zw+4(z^2+a)w'_z,
\end{equation}
which can also be obtained by differentiation with respect to~$z$ of 
the canonical fourth Painlev\' e ODE
\begin{equation}\label{PIV} 
w''_{zz}=\frac{(w'_z)^2}{2w}+\frac{3w^3}{2}+4zw^2+2(z^2+a)w+\frac{b}{2w}
\end{equation}
depending on the two parameters $a$ and $b$.) Therefore the systems of ODE
(\ref{IerPIV}) may naturally be regarded as higher analogues of the fourth Painlev\' e
equation~(\ref{PIV}). This hierarchy differs from the hierarchy of higher analogues 
of~(\ref{PIV}) considered earlier in~\cite{CJP}, \cite{CJPd}, but
can possibly be related to~it by the formulas~(\ref{chdo}).

We omit the proof of the following theorem since it contains nothing principally new compared
to the proof of Theorem~\ref{t3}.

\begin{theorem}\label{t8}
For every $n>2$ and any complex constants $t_0$, $x_0$, $a_{00},\dots,a_{(2n-1)0}$ and
 $b_{00},\dots,b_{(2n-1)0}$, the solution $(p,q)$ of the system of ODE (\ref{IerPIV}) with $t=t_0$
and initial conditions~(\ref{inipNS}) is meromorphic with respect to~$x$ on the whole complex
plane~$\mathbb{C}$. %%(\ref{Polos}).
\end{theorem}

As in the case of non-local symmetries of the KdV equation~(\ref{KdV}), our general scheme
can be used to justify the meromorphy of all solutions of initial-value problems for systems of ODE
given by the stationary parts of non-local symmetries of the NLS system~(\ref{NS}). We illustrate
this by an example, which can obviously be  generalized in the same sense as Theorems~\ref{t5}
and~\ref{t6} generalize Theorem~\ref{t4}.

Consider the system of ODE related to the master symmetry 
%\begin{equation}
\eqa{\label{symmaNS}
p'_{\tau_M}=G(t,x,p,q,v,p_1,p_2,p_3), \\ q'_{\tau_M}=\Psi(t,x,p,q,v,q_1,q_2,,q_3)}
%\end{equation}
of the complexified NLS equation~(\ref{NS}), where
%\begin{equation} 
\eqa{\label{odeG}
 G(t,x,p,q,v,p_1,p_2,p_3)=\\=2t(p_{3}-6pqp_1)+ix(p_{2}-2p^2q)+2ip_1-2ipv,}%\end{equation} 
\eqa{\label{odeP}
\Psi(t,x,p,q,v,q_1,q_2,q_3)=\\=2t(q_{3}-6pqq'_1)-ix(q_{2}-2q^2p)-2iq_1+2iqv, } %\end{split}\end{equation} 
\begin{equation}\label{xV} v'_x=pq,\end{equation}
\begin{equation}\label{tV}v'_t=i(q_1p-p_1q).\end{equation}
(The right-hand sides of this symmetry are obtained by formally applying the right-hand sides
of the recursion relations (\ref{Zhib}) to the right-hand sides of the equations of the  classical 
scaling symmetry (\ref{GNS}). The basic fact of compatibility of the system (\ref{symmaNS}) 
with the system of evolution equations~(\ref{NS}) means that the following identities hold for
the functions (\ref{odeG}) and (\ref{odeP}):
%\begin{equation}
\eqa{\label{LNSG} 
\frac{\partial G}{\partial t}+i(p_2-2p^2q)\frac{\partial G}{\partial p}-i(p_2-2p^2q)\frac{\partial G}{\partial q}+\\+i(q_1p-p_1q)\frac{\partial G}{\partial v}
-i\sum_{j=1}^3\left[\frac{\partial G}{\partial p_j}\left(\frac{d}{dx}\right)^j(p_2-2p^2q)\right]=\\=-i\frac{d^2G}{dx^2}+4ipqG+2ip^2\Psi,}
%\end{equation}
%\begin{equation}
\eqa{\label{LNSP} 
\frac{\partial \Psi}{\partial t}+i(p_2-2p^2q)\frac{\partial \Psi}{\partial p}+i(p_2-2p^2q)\frac{\partial \Psi}{\partial q}+\\+i(q_1p-p_1q)\frac{\partial \Psi}{\partial v}
+i\sum_{j=1}^3\left[\frac{\partial \Psi}{\partial p_j}\left(\frac{d}{dx}\right)^j(q_2-2q^2p)\right]=\\=i\frac{d^2\Psi}{dx^2}-4ipq\Psi-2iq^2G.}
%\end{equation}
This fact can also be verified directly.) Dividing the stationary part of the symmetry
(\ref{symmaNS}) by $2t$, we obtain the system of equations
%\begin{equation}
\eqa{\label{OdunNS} G^1(t,x,p,q,v,p_1,p_2,p_3)=0, \\ \Psi^1(t,x,p,q,v,q_1,q_2,q_3,q,v)=0,}%\end{equation}
where
%\begin{equation} 
\eqa{\nonumber G^1(t,x,p,q,p_1,p_2,p_3,v)=\\ =p_{3}-6pqp_1+\frac{ix(p_{2}-2p^2q)+2ip_1-2ipv}{2t}, 
\\ 
%\begin{equation} 
\Psi^1(t,x,p,q,q_1,q_2,q_3,v)=\\=q_{3}-6pqq_1+\frac{-ix(q_{2}-2q^2p)-2iq_1+2iqv}{2t}.}% \nonumber \end{equation} 

It follows from the identities (\ref{LNSG}), (\ref{LNSP}) that on all solutions
of the closed system of ODE (\ref{xV}), (\ref{OdunNS}) we have
%\begin{equation}
\eqa{\label{LNSo}  
\frac{dG^1}{dt}=-i\frac{d^2G^1}{dx^2}+4ipqG^1+2ip^2\Psi^1, \\ \frac{d\Psi^1}{dt}=i\frac{d^2\Psi^1}{dx^2}-4ipq\Psi^1-2iq^2G^1.}%\end{equation}
In their turn, the equalities (\ref{LNSo}) mean that the expressions
\eqa{\nonumber
S^1(t,x,p,q,v,p_1,p_2)=G^1-p_3=\\=-6pqp_1+\frac{ix(p_{2}-2p^2q)+2ip_1-2ipv}{2t}}
and
\eqa{\nonumber
S^2(t,x,p,q,v,q_1,q_2)=\Psi^1-q_3=\\=-6pqq_1+\frac{-ix(q_{2}-2q^2p)-2iq_1+2iqv}{2t}}
satisfy the following identities on all solutions of the system (\ref{xV}), (\ref{OdunNS}):
%\begin{equation*} 
\seqs {
 &\frac{\partial S^1}{\partial t}-i(p_2-2p^2q)\frac{\partial S^1}{\partial p}+\\&+i(q_2-2q^2p)\frac{\partial S^1}{\partial q}+i(q_1p-p_1q)\frac{\partial S^1}{\partial v}-\\ &-i\sum_{j=1}^2\left[\frac{\partial S^1}{\partial p_j}\left(\frac{d}{dx}\right)^j(p_2-2p^2q)-\frac{\partial S^1}{\partial q_j}\left(\frac{d}{dx}\right)^j(q_2-2q^2p)\right]=\\&=i\left(\frac{d}{dx}\right)^3(p_2-2p^2q),\\
 &\frac{\partial S^2}{\partial t}-i(p_2-2p^2q)\frac{\partial S^2}{\partial p}+\\&+i(p_2-2p^2q)\frac{\partial S^2}{\partial q}+i(q_1p-p_1q)\frac{\partial S^2}{\partial v}-\\&-i\sum_{j=1}^2\left[\frac{\partial S^2}{\partial p_j}\left(\frac{d}{dx}\right)^j(p_2-2p^2q)-\frac{\partial S^2}{\partial q_j}\left(\frac{d}{dx}\right)^j(q_2-2q^2p)\right]=\\&=-i\left(\frac{d}{dx}\right)^3(q_2-2q^2p).}%\end{equation*}
Moreover, the function $S^3(p,q)=pq$ satisfies the obvious equality
$$
\frac{\partial S^3}{\partial t}-i(p_2-2p^2q)\frac{\partial S^3}{\partial p}+i(q_2-2q^2p)\frac{\partial S^2}{\partial q}=i\frac{d}{dx}(q_1p-p_1q).
$$
By Definition~\ref{d2}, these three equalities mean that the system of three non-linear ODE
(\ref{xV}), (\ref{OdunNS}) determines an invariant manifold of the system of evolution equations
(\ref{NS}), (\ref{tV}). Combining this with Theorem~\ref{t1} and the result of~\cite{Dom} about
the global meromorphy of all local holomorphic solutions of the complexified NLS equation
(\ref{NS}), we deduce the following theorem by our standard scheme.

\begin{theorem}\label{t9}
For any values of the complex constants $t_0\neq 0 $, $x_0$, $a_i$ and $b_i$ $(i=0,1, 2)$,
the solution of the system of ODE (\ref{xV}), (\ref{OdunNS}) with $t=t_0$ and initial conditions
$$
p_i(t_0,x_0)=a_i, \qquad q_i(t_0,x_0)=b_i \qquad (i=0,1, 2)
$$
is meromorphic with respect to $x$ on the whole complex plane~$\mathbb{C}$.
\end{theorem}

\begin{remark}\label{r11}
{\rm The set of solutions of the stationary part of the symmetry (\ref{symmaNS}) 
with the reduction  $p=-q^*$ contains the even (in~$x$) simultaneous solutions
of the equation 
\begin{equation}\label{samof}
2t(q_{3}-6pqq_1)-ix(q_{2}-2q^2p)-2iq_1-2iq\int_0^x|q|^2(t,\zeta)d\zeta 
\end{equation}
and the focusing NLS equation $q'_t=i(q_{xx}+2|q|^2)$ which universally describe 
the corrective influence of small dispersion on the self-focusing of non-linear
geometric optic approximations \cite{samof} in spatially one-dimensional problems with 
sufficiently small non-linearities. These simultaneous solutions of (\ref{samof}) and the focusing
NLS equation also arise in the description of rogue waves of infinite order~\cite{kilin}. 
As indicated in~\cite{kilin}, they
satisfy an ODE which is a particular case of the second term of the third Painlev\' e
hierarchy wriiten out in~\cite{sakka}. See also \cite{denb}--\cite{lizh} for the properties of these solutions and
their generalizations.}
\end{remark}

\section{\label{sec:level5}Systems of equations of Painlev\' e type that determine invariant manifolds 
of the Sawada--Kotera equation, and meromorphy of their solutions} 

It is known that the stationary parts of symmetries determine invariant manifolds
of systems of evolution equations. But the converse is not true in general. A striking
example of a hierarchy of such invariant manifolds was given by Bagderina~\cite{Bagd} in
the case of the Sawada--Kotera equation (\ref{SK}). The relation of these invariant manifolds
to symmetries of~(\ref{SK}) is highly non-trivial. The hierarchy is given by the ODE 
\begin{equation}\label{Bagd}
U_i(u,\dots,u_{2i})=0,
\end{equation}
where $U_i(u,\dots,u_{2i})$ is a sequence of polynomials
\eqa
{\nonumber U_0=u,\qquad U_1=u_2-6u^2,\\
U_3=u_6 -6(6uu_4 + 10u_1u_3 + 5(u_2)^2) + 360(u^2u_2 + u(u_1)^2 - u^4),\\ \dots.}
The formula 
\begin{equation} U_{i+1}=\left (\frac{d^2}{dx^2}\right)J_i-uJ_i,\label{Bagdi}\end{equation}
expresses them in terms of the polynomials $J_i(u,\dots,u_{2i})$ with
\eqa{\nonumber
J_{-1}=-1/6, \quad J_0=u,\quad J_2=u_4 - 30uu_2 + 60u^3, \\
J_3=u_6 -12(uu_4 + u_1u_3 + (u_2)^2) + 252 (2u^2u_2 + u(u_1)^2 - 2u^4),\\ \dots,}
which determine a familiar \cite{Fuo} infinite series of symmetries of the SK equation (\ref{SK}):
\begin{equation} u'_{\tau_i}=f_i=\left (\frac{d}{dx}\right)J_i.\label{Bagdi}\end{equation}
Note that $U_{3i}=J_{3i+1}=0$ for every $i>0$. These symmetries satisfy the recursion
relation $f_{i+3}=Rf_i$ with
\eqa{\nonumber
R= \left [\frac{d^2}{dx^2} - 24u - 12u_1\left(\frac{d}{dx}\right)^{-1}\right]
\left[\frac{d^2}{dx^2} -\right.\\ \left. -6u\right]\frac{d}{dx}\left[\frac{d^2}{dx^2} - 6u\right]\left(\frac{d}{dx}\right)^{-1}.}
This relation was found in~\cite{Sok}. 

We generalize Bagderina's construction to the case of a system of non-autonomous ODE.
Consider the simplest non-autonomous symmetry of the SK equation~(\ref{SK}), 
the classical scaling symmetry 
\eqa{\nonumber
u_{\tau_r}=5tu_t+xu_x+2u=\\=
5t[u_5-30uu_3-30u_1u_2+180u^2u_1]+xu_x+2u.}
(The stationary part  
\begin{equation} \label{rstS} 5t(u_5-30uu_3-30u_1u_2+180u^2u_1)+xu_1+2u=0\end{equation}
of this symmetry determines self-similar solutions of the Sawada--Kotera equation by 
the change of variables
$$
z=\frac{x}{t^{1/5}}, \qquad u=\frac{w(z)}{t^{2/5}}.
$$
This change of variables reduces the SK equation (\ref{SK}) and the ODE (\ref{rstS})
to the $t$-independent non-linear ODE
$$
5(w_5-30ww_3-30w_1w_2+180w^2w_1)+zw_1+2w=0
$$
on the function $w(z)$, where $w_j=\frac{d^jw}{dz^j}$.) The derivative in $x$ of the function
\begin{equation} \label{formB} 5t[u_4-30uu_2+60u^3]+{xu+v},\end{equation}
where \begin{equation}\label{PervS}v_x=u,\end{equation} 
coincides with the right-hand side of the scaling symmetry. Formally acting on (\ref{formB})
by the differential operator in the right-hand side of~(\ref{Bagdi}) and dividing the result by~$5t$,
we obtain a function
\eqa{\nonumber
 U_r=u_6-6(6uu_4+10u_1u_3+5u_2^2)+360(u^2u_2+uu_1^2-u^4)-\\-\frac{x(u_2-6u^2)+3u_1-6uv}{5t}.
}
A direct verification shows that the system of ODE consisting of the equations (\ref{PervS}) 
and the non-autonomous ODE
\begin{equation}\label{inaB} U_r=0 \end{equation}
determines an invariant manifold (in the sense of Definition~\ref{d2}) of the evolutionary system
consisting of the equation
\begin{equation}\label{evov}v_t=u_4-30uu_2+60u^3\end{equation}
and the SK equation (\ref{SK}).

The following theorem is proved along the same lines as in Theorems~\ref{t3}--\ref{t9}.

\begin{theorem}\label{t10}
For any complex numbers $t_0\neq 0$, $x_0$, $a$ and $b_i$ $(i=0,\dots, 4)$, 
the solution of the system
of ODE (\ref{PervS}), (\ref{inaB}) with $t=t_0$ and initial conditions~(\ref{inisn}) is meromorphic
with respect to~$x$ on the whole complex plane~$\mathbb{C}$.
\end{theorem}

\begin{proof}. In~\cite{Dodd}, the SK equation (\ref{SK}) was represented in the Lax form
$L_t=[A,L]$ with differential operators
$$
 L =D^3+uD, 
A=-9D^5+80uD^3+80u_1D^2-(180u^2-60u_2)D
$$
of coprime orders, where $D=\frac{d}{dx}$. By the main theorem in~\cite{Domok} 
(with an irrelevant inaccuracy corrected in~\cite{Domod}), all the coefficients of~$L$ 
(in our case, $u(t,x)$) for any local holomorphic solution extend to globally meromorphic 
functions of the spatial variable~$x$. Since the system of ODE (\ref{PervS}), (\ref{inaB}) 
determines an invariant manifold of the system of evolution equations (\ref{evov}), (\ref{SK}),
Theorem~\ref{t1} yields a simultaneous solution of these two systems
 with initial conditions~(\ref{inisn}) in a neighborhood of the point $t=t_0, x=x_0$. The rest
of the proof is clear.
\end{proof}

\section{Conclusion}

We have demonstrated on a number of examples that the problem of rigorously proving
the global meromorphy of solutions of many non-linear ODE or their systems
can be solved easily and uniformly by using the general scheme described in the proofs 
of Theorems~\ref{t3}--\ref{t10}. Here are the key assumptions under which this scheme works.

1) The system of non-linear ODE under study is equivalent to a system of ODE 
of the form~(\ref{invm}) such that its right-hand sides are globally analytic in the independent 
variable~$x$ and it
determines an invariant manifold of some evolutionary system of the form~(\ref{sise}). 
Then Theorem~\ref{t1} guarantees the existence of a simultaneous holomorphic
solution in a neighborhood of any given point in~$\mathbb{C}^2$.

2) The solutions of (\ref{sise}) are expressible 
in terms of solutions of an evolutionary system (usually integrable by inverse scattering transform)
with the following meromorphic extension property. Any local holomorphic solution can be extended
globally in $x$ (and locally in other independent variables) to a strip in~$\mathbb{C}^2$.

The proof of this property is the main non-trivial point of our scheme. It is currently
unavailable for some integrable systems, say, for the Kadomtsev--Petviashvili equation.
However, it is available for the systems (\ref{KdV})--(\ref{SK}) and their hierarchies
as well as for many other IST-integrable systems~\cite{Dom}--\cite{Domp} including the
Boussinesq equation
\begin{equation} \label{buss} 
u''_{tt}=au''''_{xxxx} +buu''_{xx} +b(u'_x)^2 
\end{equation}
with any non-zero complex constants $a,b$ (see~\cite{ Domb}) and the vector and matrix
versions of the complexified NLS equation~\cite{Domk}. Our scheme gives a
simple way to prove the meromorphy of solutions of non-autonomous systems of ODE
that determine invariant manifolds of these evolutionary equations. In particular, it is plausible
that such a proof can be performed for Cosgrove's equation 
\eqa{\nonumber
 v''''_{xxxx}+12vv''_{xx}+6(v'_x)^2+ \frac{32}{3}v^3+a(vv''_{xx}+(v'_x)^2) - bx=0,
 \\ a,b= \mathrm{const},}
which is related to solutions of~(\ref{buss})
\cite{Cosg} -- \cite{Corp}, as well as for some matrix analogues of the second Painlev\' e
equation described in~\cite {Adsok} and their hierarchies.

\begin{acknowledgments}
The first named author was supported by the Russian Foundation for Basic Research,
grant no,~19-01-00474.
The authors are grateful to V. V. Sokolov and S. V. Khabirov for bringing the paper \cite{Bagd}
to our attention. We also profited from the consultations of O. V. Kaptsov and V. V. Sokolov 
concerning the existence of simultaneous solutions of systems of evolutionary equations 
and ODE determining their invariant manifolds.

\end{acknowledgments}

%%%%%%%%%%%%%%%%

\end{document}